\documentclass[10pt,a4paper]{article}
\usepackage{amssymb}
\usepackage{enumitem}
\usepackage{amsmath}
\usepackage{here}
\usepackage{graphicx, xcolor}
\usepackage{tikz}
\usepackage{tikzsymbols}
\usepackage{subcaption}

\usepackage[margin=1in]{geometry}
\usetikzlibrary{shapes.geometric, arrows.meta}
\usetikzlibrary{decorations.pathreplacing}

\newif\ifdraft \drafttrue   

\ifdraft
\newcommand{\shinocom}[1]{\textcolor{magenta}{[(篠原)#1]}}
\newcommand{\yoshicom}[1]{\textcolor{red}{[(吉仲)#1]}}
\newcommand{\usercom}[1]{\textcolor{blue}{[(user)#1]}}
\newcommand{\todo}[1]{{\color{red}{[ToDo: #1]}}}
\else
\newcommand{\shinocom}[1]{}
\newcommand{\yoshicom}[1]{}
\newcommand{\usercom}[1]{}
\newcommand{\todo}[1]{}
\fi

\usepackage{amsthm}
\usepackage{bm}
\usepackage{url}                               

\newcommand{\tbf}[1]{\textbf{#1}}
\newcommand{\tsc}[1]{\textsc{#1}}
\newcommand{\mcal}[1]{\mathcal{#1}}
\newcommand{\mrm}[1]{\mathrm{#1}}

\newcommand{\mbb}[1]{\mathbb{#1}}
\newcommand{\mbf}[1]{\mathbf{#1}}

\theoremstyle{definition}
\newtheorem{definition}{Definition}

\theoremstyle{plain}
\newtheorem{theorem}[definition]{Theorem}
\newtheorem{corollary}[definition]{Corollary}
\newtheorem{lemma}[definition]{Lemma}

\definecolor{dgreen}{rgb}{0,.5,0}
\definecolor{dyellow}{rgb}{.7,.7,0}

\newcommand{\qr}[1]{\textcolor{red}{\psg{R}}^{#1}}
\newcommand{\qg}[1]{\textcolor{dgreen}{\psg{G}}^{#1}}
\newcommand{\qb}[1]{\textcolor{blue}{\psg{B}}^{#1}}
\newcommand{\qp}[1]{\textcolor{violet}{\psg{P}}^{#1}}
\newcommand{\qy}[1]{\textcolor{dyellow}{\psg{Y}}^{#1}}
\newcommand{\psg}[1]{\mbf{#1}}

\usepackage{titlesec}

\makeatletter

\def\ps@headings{
	\def\@oddhead{		}
	\def\@evenhead{}
	\def\@evenfoot{\hfil \,\,\thepage\hfill}
	\def\@oddfoot{\hfil \,\,\thepage\hfil}
}

\def\ps@titleheadings{
	\def\@oddhead{}
	\def\@evenhead{}
	\def\@evenfoot{\hfil \,\,\thepage\hfill}
	\def\@oddfoot{\hfil \,\,\thepage\hfil}
}
\makeatother

\title{
	 \textbf{BusOut is NP-complete}
	}
\author{
	Takehiro Ishibashi\thanks{Tohoku University} \and
	Ryo Yoshinaka\footnotemark[1] \and
	Ayumi Shinohara\footnotemark[1]
}
\date{\empty}

\begin{document}
\maketitle

\begin{abstract}
  This study examines the computational complexity of the decision problem modeled on the smartphone game \textit{Bus Out}.
  The objective of the game is to load all the passengers in a queue onto appropriate buses using a limited number of bus parking spots by selecting and dispatching the buses on a map.
  We show that the problem is NP-complete, even for highly restricted instances.
  We also show that it is hard to approximate the minimum number of parking spots needed to solve a given instance.
\end{abstract}

\section{Introduction}

 The study of computational complexity in puzzles and games is an active research field in theoretical computer science~\cite{Uehara23}.
 Particularly, many puzzles, i.e., single-player games, have been shown to be NP-complete.
 For example, long-loved games such as pencil puzzles like Sudoku~\cite{sudoku}, as well as Minesweeper~\cite{minesweeper} and Tetris~\cite{BreukelaarDHHKL04}, have also been studied, and it has been shown that all of them are NP-complete.
%
%

This paper focuses on \textit{Bus Out}, a popular single-player smartphone game developed by iKame Games.
The game has gained significant attention in recent years. 
The objective is to dispatch buses to the station, and load and carry all waiting passengers there.
Each passenger is assigned a color and each bus is assigned a color and a capacity.
Each bus can accommodate only passengers of the same color up to its capacity.
Initially, the station has passengers but no buses.
Buses are caught in traffic.
A bus cannot get out of traffic if there are other buses blocking its intended path.
The player selects a bus that has no other buses in its facing direction, and dispatches it to the station.
The station has a limited number of bus parking spots, each of which accommodates a single bus.
All passengers in the station are arranged in a line and must board in order.
The passenger queue is strict:
if the first passenger cannot find a matching-colored bus, no one is allowed to board, even if suitable buses are available for others further back in the line.
When (and only when) a bus becomes full, it immediately departs the station, freeing up its parking spot.
The game is successfully cleared when all buses and passengers have departed.
If a deadlock occurs, the game cannot be completed.

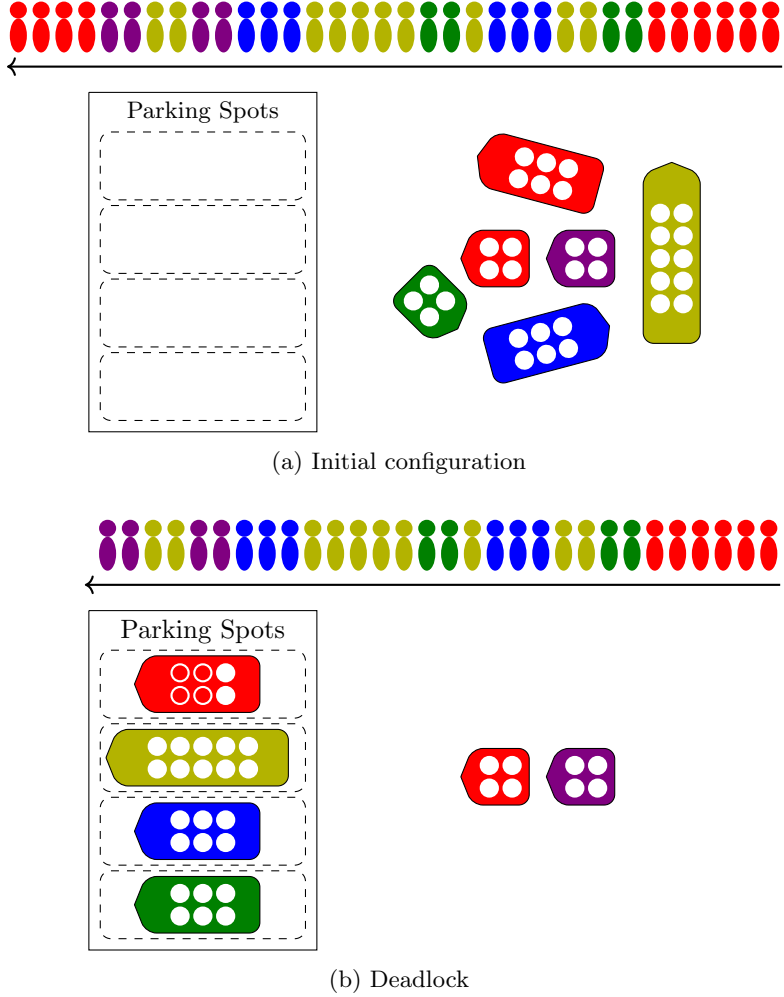
\begin{figure}[t]
    \centering
    \begin{subfigure}{0.97\columnwidth}
        \centering
        \newcommand{\human}[2]{
    \fill [#1] (#2, 0) circle (0.15);
    \fill [#1] (#2, -0.45) circle [x radius = 0.15, y radius = 0.3];
}
\newcommand{\smallbus}[4]{
    \draw [fill = #1, rotate around={#4:(#2, #3)}] (#2, #3 + 0.7) [rounded corners] -- (#2 - 0.5, #3 + 0.5) [rounded corners] -- (#2 - 0.5, #3 - 0.5) [rounded corners] -- (#2 + 0.5, #3 - 0.5) [rounded corners] -- (#2 + 0.5, #3 + 0.5) -- (#2, #3 + 0.7);
    \foreach \x in {0.2, -0.2}
        \foreach \y in {-0.2, 0.2}
            \draw [white, thick, fill = white, rotate around={#4:(#2, #3)}] (#2 + \x, #3 + \y) circle [radius = 0.15];
}
\newcommand{\mediumbus}[4]{
    \draw [fill = #1, rotate around={#4:(#2, #3)}] (#2, #3 + 1.2) [rounded corners] -- (#2 - 0.5, #3 + 1) [rounded corners] -- (#2 - 0.5, #3 - 1) [rounded corners] -- (#2 + 0.5, #3 - 1) [rounded corners] -- (#2 + 0.5, #3 + 1) -- (#2, #3 + 1.2);
    \foreach \x in {0.2, -0.2}
        \foreach \y in {-0.4, 0, 0.4}
            \draw [white, thick, fill = white, rotate around={#4:(#2, #3)}] (#2 + \x, #3 + \y) circle [radius = 0.15];
}
\newcommand{\largebus}[4]{
    \draw [fill = #1, rotate around={#4:(#2, #3)}] (#2, #3 + 1.7) [rounded corners] -- (#2 - 0.5, #3 + 1.5) [rounded corners] -- (#2 - 0.5, #3 - 1.5) [rounded corners] -- (#2 + 0.5, #3 - 1.5) [rounded corners] -- (#2 + 0.5, #3 + 1.5) -- (#2, #3 + 1.7);
    \foreach \x in {0.2, -0.2}
        \foreach \y in {-0.8, -0.4, 0, 0.4, 0.8}
            \draw [white, thick, fill = white, rotate around={#4:(#2, #3)}] (#2 + \x, #3 + \y) circle [radius = 0.15];
} 

\begin{tikzpicture}[scale=0.75]
    \foreach \x/\y in {0/red, 0.4/red, 0.8/red, 1.2/red, 1.6/violet, 2/violet, 2.4/dyellow, 2.8/dyellow, 3.2/violet, 3.6/violet, 4/blue, 4.4/blue, 4.8/blue, 5.2/dyellow, 5.6/dyellow, 6/dyellow, 6.4/dyellow, 6.8/dyellow, 7.2/dgreen, 7.6/dgreen, 8/dyellow, 8.4/blue, 8.8/blue, 9.2/blue, 9.6/dyellow, 10/dyellow, 10.4/dgreen, 10.8/dgreen, 11.2/red, 11.6/red, 12/red, 12.4/red, 12.8/red, 13.2/red}
        \human{\y}{\x};
    \draw[<-,thick] (-0.2,-1) -- (13.4,-1);
\end{tikzpicture}
\vspace{0.5em}

\begin{tikzpicture}[baseline = -12em,scale=0.75]
    \draw (0, 0) -- (0, -6) -- (4, -6) -- (4, 0) -- (0, 0);
    \node [below] (title) at (2, 0) {\small Parking Spots};
    \foreach \x in {0.7, 2, 3.3, 4.6}
        \draw [dashed, rounded corners] (0.2, -\x) -- (3.8, -\x) -- (3.8, -\x - 1.2) -- (0.2, -\x - 1.2) -- cycle;
\end{tikzpicture}  
\hspace{2em}
\begin{tikzpicture}[baseline = -10em,scale=0.75]
    \smallbus{red}{2}{-2}{90};
    \smallbus{violet}{3.5}{-2}{90};
    \smallbus{dgreen}{0.75}{-2.75}{-135};
    \mediumbus{blue}{2.75}{-3.5}{-75};
    \largebus{dyellow}{5}{-2}{0};
    \mediumbus{red}{2.75}{-0.5}{75};
\end{tikzpicture}
        \caption{Initial configuration}
        \label{fig:example_init}
    \end{subfigure}
    \\[1.5em]
    \begin{subfigure}{0.97\columnwidth}
        \centering
        \newcommand{\human}[2]{
    \fill [#1] (#2, 0) circle (0.15);
    \fill [#1] (#2, -0.45) circle [x radius = 0.15, y radius = 0.3];
}
\newcommand{\smallbus}[4]{
    \draw [fill = #1, rotate around={#4:(#2, #3)}] (#2, #3 + 0.7) [rounded corners] -- (#2 - 0.5, #3 + 0.5) [rounded corners] -- (#2 - 0.5, #3 - 0.5) [rounded corners] -- (#2 + 0.5, #3 - 0.5) [rounded corners] -- (#2 + 0.5, #3 + 0.5) -- (#2, #3 + 0.7);
    \foreach \x in {0.2, -0.2}
        \foreach \y in {-0.2, 0.2}
            \draw [white, thick, fill = white, rotate around={#4:(#2, #3)}] (#2 + \x, #3 + \y) circle [radius = 0.15];
}
\newcommand{\mediumbus}[4]{
    \draw [fill = #1, rotate around={#4:(#2, #3)}] (#2, #3 + 1.2) [rounded corners] -- (#2 - 0.5, #3 + 1) [rounded corners] -- (#2 - 0.5, #3 - 1) [rounded corners] -- (#2 + 0.5, #3 - 1) [rounded corners] -- (#2 + 0.5, #3 + 1) -- (#2, #3 + 1.2);
    \foreach \x in {0.2, -0.2}
        \foreach \y in {-0.4, 0, 0.4}
            \draw [white, thick, fill = white, rotate around={#4:(#2, #3)}] (#2 + \x, #3 + \y) circle [radius = 0.15];
}
\newcommand{\largebus}[4]{
    \draw [fill = #1, rotate around={#4:(#2, #3)}] (#2, #3 + 1.7) [rounded corners] -- (#2 - 0.5, #3 + 1.5) [rounded corners] -- (#2 - 0.5, #3 - 1.5) [rounded corners] -- (#2 + 0.5, #3 - 1.5) [rounded corners] -- (#2 + 0.5, #3 + 1.5) -- (#2, #3 + 1.7);
    \foreach \x in {0.2, -0.2}
        \foreach \y in {-0.8, -0.4, 0, 0.4, 0.8}
            \draw [white, thick, fill = white, rotate around={#4:(#2, #3)}] (#2 + \x, #3 + \y) circle [radius = 0.15];
}
\newcommand{\fillingmediumbus}[4]{
    \draw [fill = #1, rotate around={#4:(#2, #3)}] (#2, #3 + 1.2) [rounded corners] -- (#2 - 0.5, #3 + 1) [rounded corners] -- (#2 - 0.5, #3 - 1) [rounded corners] -- (#2 + 0.5, #3 - 1) [rounded corners] -- (#2 + 0.5, #3 + 1) -- (#2, #3 + 1.2);
    \draw [white, thick, rotate around={#4:(#2, #3)}] (#2 + 0.2, #3 + 0.4) circle [radius = 0.15];
    \draw [white, thick, rotate around={#4:(#2, #3)}] (#2 - 0.2, #3 + 0.4) circle [radius = 0.15];
    \draw [white, thick, rotate around={#4:(#2, #3)}] (#2 + 0.2, #3) circle [radius = 0.15];
    \draw [white, thick, rotate around={#4:(#2, #3)}] (#2 - 0.2, #3) circle [radius = 0.15];
    \draw [white, fill = white, thick, rotate around={#4:(#2, #3)}] (#2 + 0.2, #3 - 0.4) circle [radius = 0.15];
    \draw [white, fill = white, thick, rotate around={#4:(#2, #3)}] (#2 - 0.2, #3 - 0.4) circle [radius = 0.15];
}

\newcommand{\dummy}[4]{
    \ifnum#1=1
        \draw [white, fill = white, rotate around={#4:(#2, #3)}] (#2, #3 + 0.7) [rounded corners] -- (#2 - 0.5, #3 + 0.5) [rounded corners] -- (#2 - 0.5, #3 - 0.5) [rounded corners] -- (#2 + 0.5, #3 - 0.5) [rounded corners] -- (#2 + 0.5, #3 + 0.5) -- (#2, #3 + 0.7);
    \else\ifnum#1=2
        \draw [white, fill = white, rotate around={#4:(#2, #3)}] (#2, #3 + 1.2) [rounded corners] -- (#2 - 0.5, #3 + 1) [rounded corners] -- (#2 - 0.5, #3 - 1) [rounded corners] -- (#2 + 0.5, #3 - 1) [rounded corners] -- (#2 + 0.5, #3 + 1) -- (#2, #3 + 1.2);
    \else\ifnum#1=3
        \draw [white, fill = white, rotate around={#4:(#2, #3)}] (#2, #3 + 1.7) [rounded corners] -- (#2 - 0.5, #3 + 1.5) [rounded corners] -- (#2 - 0.5, #3 - 1.5) [rounded corners] -- (#2 + 0.5, #3 - 1.5) [rounded corners] -- (#2 + 0.5, #3 + 1.5) -- (#2, #3 + 1.7);
    \fi\fi\fi
}

\begin{tikzpicture}[scale=0.75]
    \foreach \x/\y in {0/white, 0.4/white, 0.8/white, 1.2/white, 1.6/violet, 2/violet, 2.4/dyellow, 2.8/dyellow, 3.2/violet, 3.6/violet, 4/blue, 4.4/blue, 4.8/blue, 5.2/dyellow, 5.6/dyellow, 6/dyellow, 6.4/dyellow, 6.8/dyellow, 7.2/dgreen, 7.6/dgreen, 8/dyellow, 8.4/blue, 8.8/blue, 9.2/blue, 9.6/dyellow, 10/dyellow, 10.4/dgreen, 10.8/dgreen, 11.2/red, 11.6/red, 12/red, 12.4/red, 12.8/red, 13.2/red}
        \human{\y}{\x};
    \draw[<-,thick] (1.2,-1) -- (13.4,-1);
\end{tikzpicture}
\vspace{0.5em}

\begin{tikzpicture}[baseline = -12em,scale=0.75]
    \draw (0, 0) -- (0, -6) -- (4, -6) -- (4, 0) -- (0, 0);
    \node [below] (title) at (2, 0) {Parking Spots};
    \foreach \x in {0.7, 2, 3.3, 4.6}
        \draw [dashed, rounded corners] (0.2, -\x) -- (3.8, -\x) -- (3.8, -\x - 1.2) -- (0.2, -\x - 1.2) -- cycle;
    \fillingmediumbus{red}{2}{-1.3}{90}
    \largebus{dyellow}{2}{-2.6}{90}
    \mediumbus{blue}{2}{-3.9}{90}
    \mediumbus{dgreen}{2}{-5.2}{90}
\end{tikzpicture}  
\hspace{2em}
\begin{tikzpicture}[baseline = -10em,scale=0.75]
    \smallbus{red}{2}{-2}{90};
    \smallbus{violet}{3.5}{-2}{90};
    \dummy{1}{0.75}{-2.75}{-135};
    \dummy{2}{2.75}{-3.5}{-75};
    \dummy{3}{5}{-2}{0};
    \dummy{2}{2.75}{-0.5}{75};
\end{tikzpicture}
        \caption{Deadlock}
        \label{fig:example_wrong}
    \end{subfigure}
    \caption{Example scenario of Bus Out. (a) initial configuration; (b) deadlock caused by a misstep.}
    \label{fig:example_scenario}
\end{figure}

Consider the scenario in Figure~\ref{fig:example_init}, where four parking spots are available and bus capacities are four, six, and ten.
At the front of the queue are four red passengers, and a red bus with capacity six is immediately dispatchable.
However, if we naively dispatch this bus first, it will occupy a spot with two vacant seats.
To make room for the purple bus---needed to board the succeeding purple passengers---all other buses, effectively blocking its path, must first be dispatched.
With only three spots remaining, this leads to an unavoidable deadlock (Figure~\ref{fig:example_wrong}).

Instead, let us first dispatch the yellow, blue, and green buses.
They occupy three spots with no passengers.
Then, dispatching the red bus with capacity four allows the four red passengers to board and depart.
This clears the way for the purple bus and passengers.
Finally, the yellow, blue, and green buses fill and depart, followed by the remaining six red passengers using the red bus with capacity six.
This sequence leads to a successful completion of the scenario.

In this paper, by formally defining the game as a computational decision problem, we show that it is NP-complete even for highly restricted cases.
Namely, the problem is NP-hard when the station has a single parking spot, buses have two colors and a unique capacity.
On the other hand, the monochrome version of the problem becomes trivial regardless of the other parameters.
Even with no traffic congestion, the problem remains NP-hard if possible bus capacities are not fixed.

An optimization version of the game can also be considered that asks the number of necessary parking spots to clear the game.
We show that a polynomial-time approximation with any constant ratio is impossible unless $\mrm{P} = \mrm{NP}$. 

\section{Formalization of \textit{Bus Out}}

We work with a fixed set of colors.
The \emph{passenger queue} is formalized as a sequence of colors.
A parking spot is either empty, denoted as $\varepsilon$, or occupied by a bus, represented as a pair $(x,k)$ of a color $x$ and a number $k$ of remaining seats.
The bus station has a limited number of parking spots.
The buses in traffic are represented in a labeled directed graph, called a \emph{congestion graph}, where each vertex represents a bus with a label of its color and capacity, and we have a direct edge between two buses if one blocks the other.
We call a vertex (or bus) in the congestion graph \emph{free} if its out-degree is zero.
A \emph{configuration} is a triple $(G,Q,S)$ of a congestion graph $G$, a passenger queue $Q$, and a spot occupancy state $S$. 
It is called \emph{empty} if $G$ and $Q$ are empty and $S$ consists of empty spots.

The initial configuration illustrated in Figure~\ref{fig:example_init} is formalized as
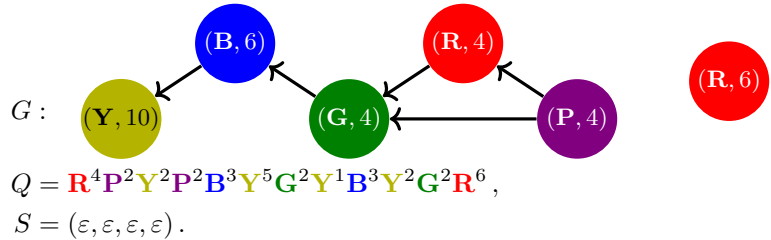
\begin{align*}
    G &: \quad
    \begin{tikzpicture}[every node/.style={fill,circle,minimum size=30pt, inner sep=0pt, font=\small},baseline = 0pt]
        \node (A) [color = dyellow, text = black] at (0, 0) {$(\psg{Y},10)$};
        \node (B) [color = blue, text = white] at (1.5, 1) {$(\psg{B},6)$};
        \node (C) [color = dgreen, text = white] at (3, 0) {$(\psg{G},4)$};
        \node (D) [color = red, text = white] at (4.5, 1) {$(\psg{R},4)$};
        \node (E) [color = violet, text = white] at (6, 0) {$(\psg{P},4)$};
        \node (F) [color = red, text = white] at (8, 0.5) {$(\psg{R},6)$};
        \draw [<-, very thick] (A) -- (B);
        \draw [<-, very thick] (B) -- (C);
        \draw [<-, very thick] (C) -- (D);
        \draw [<-, very thick] (D) -- (E);
        \draw [<-, very thick] (C) -- (E);
    \end{tikzpicture}
    \\
    Q & = \qr{4}\qp{2}\qy{2}\qp{2}\qb{3}\qy{5}\qg{2}\qy{1}\qb{3}\qy{2}\qg{2}\qr{6}\,,
    \\
    S &= (\varepsilon, \varepsilon, \varepsilon, \varepsilon)\,.
\end{align*}
A passenger queue is written as $x_1 \dots x_k$ rather than $(x_1,\dots,x_k)$.
If the same color (or color sequence) $x$ repeats $n$ times, it is denoted as $x^n$.
Similarly, the configuration in Figure~\ref{fig:example_wrong} is 
\begin{align*}
    G &:\quad
    \begin{tikzpicture}[every node/.style={fill,circle,minimum size=30pt, inner sep=0pt, font=\small},baseline = 0pt]
        \node (D) [color = red, text = white] at (2, 0) {$(\psg{R},4)$};
        \node (E) [color = violet, text = white] at (4, 0) {$(\psg{P},4)$};
        \draw [<-, very thick] (D) -- (E);
    \end{tikzpicture}
    \\
    Q & = \qp{2}\qy{2}\qp{2}\qb{3}\qy{5}\qg{2}\qy{1}\qb{3}\qy{2}\qg{2}\qr{6}
    \,,\\
    S &= (({\qr{}},2), ({\qy{}}, 10), ({\qb{}}, 6), (\qg{},4))
\,.\end{align*}

A configuration $(G,Q,S)$ can transition to $(G',Q',S')$ if one of the following holds:
\begin{itemize}
    \item $Q'=Q$, $G'$ is $G$ minus a free vertex labeled with $(x,k)$, and $S'$ is obtained from $S$ by replacing an element $\varepsilon$ with $(x,k)$;
    \item $G'=G$, $Q'$ is $Q$ with the first element $x$ removed, and $S'$ is obtained from $S$ by replacing an element $(x,k+1)$ with $(x,k)$ for some $k \ge 1$, or $(x,1)$ with $\varepsilon$.
\end{itemize}
The player has a control on a transition of the former type, while the latter takes place automatically.\footnote{When the station has two or more buses of the same color as the first passenger, the leftmost bus is chosen in the actual smartphone game, whereas the bus with the least available seats is a reasonable choice for the player. However, this difference does not matter in the following discussions of this paper.}
Whenever a transition of the second type is applicable, it takes place before the player selects a free vertex to cause the first type transition.
A configuration is \emph{solvable} if one can make the configuration empty by a sequence of transitions.
If a nonempty configuration allows no transition, it is a \emph{deadlock}.
We say that a  configuration $(G,Q,S)$ \emph{eligible} just in the case where $G$ contains no cycle, and the total capacity of the buses of each color matches the number of passengers of the same color.
It is evident that eligibility is an invariant property under transitions, and that a configuration $(G,Q,S)$ is solvable only if it is eligible.
Thus, throughout this paper, we consider only eligible instances.
The decision problem $\mathbf{BusOut}$ asks whether an input (eligible) configuration is solvable.
In addition, for simplicity, we consider only the empty spot state in initial configurations.
The empty spot state is denoted as $\varepsilon^s$, where $s$ is the number of spots in the station.

In the actual smartphone game, the set of colors and the possible capacities of buses are fixed.
The number of parking spots is fixed as well, unless you pay.
Accordingly, it is natural to consider classes of configurations defined by those parameters.
Let $s \in \mbb{N}_+$, $c \in \mbb{N}_+$, and $V \subseteq \mbb{N}_+$ with $V \neq \emptyset$, where $\mbb{N}_+$ denotes the set of positive integers.
By $\mcal{B}(s,c,V)$, we denote the class of configurations where the number of bus parking spots is $s$, the number of distinct colors for buses and passengers is at most $c$, and bus capacities are restricted to elements of $V$.
For example, all game configurations appearing in \textit{Bus Out} by iKame Games belong to $\mcal{B}(4, 8, \{4, 6, 10\})$, unless you pay.
Note that $\mcal{B}(s,c,V) \subseteq \mcal{B}(s,c',V')$ if $c \le c'$ and $V \subseteq V'$, whereas $\mcal{B}(s,c,V) \cap \mcal{B}(s',c,V) = \emptyset$ if $s \neq s'$.
Using $c$ colors is a special case of using $c' > c$ colors, but having $s$ spots is not a special case of having $s' > s$ spots.
For a class $\mcal{B}$ of configurations, we define $\mathbf{BusOut}(\mcal{B})$ to be the problem whose instances are from $\mcal{B}$.

\section{The computational complexity of $\mathbf{BusOut}$}
It is obvious that $\mbf{BusOut}$ belongs to NP, since the congestion graph and the passenger queue are monotonically shrunk by transitions.
In this section, we discuss conditions for the problem to be NP-complete and to belong to P.

We first show that the monochrome instances are trivially solvable.
\begin{theorem}\label{thm:monochrome}
    Every instance of\/ $\mbf{BusOut}(\mcal{B}(s,1,V))$ is solvable for any $s$ and $V$.
\end{theorem}
\begin{proof}
    If there is a bus in the station, the second type of a transition immediately takes place.
    Otherwise, a transition of the first type is possible, unless it is already the empty configuration.
\end{proof}

While the monochrome version of \tbf{BusOut} is trivial, the following theorem shows that \tbf{BusOut} with two colors is hard even when the other parameters are very much restricted.
\begin{theorem}\label{theorem:121_comp}
    The problem $\mbf{BusOut}(\mcal{B}(1,2,\{1\}))$ is NP-complete.
\end{theorem}
\begin{proof}
The NP-hard problem that we use for our reduction is the \tbf{3-Partition} problem~\cite{GareyJohnson75}.
Given a multiset $M = \{a_1, a_2, \dots, a_{3n}\}$ of $3n$ positive integers as input, the problem asks whether it is possible to partition the elements of $M$ into $n$ subsets, each containing exactly three elements, such that the sum of the elements in each subset is equal to $T = \frac{1}{n} \sum_{i = 1}^{3n} a_i$.
It is known that the problem remains NP-hard even when each element $a_i$ of $M$ satisfies the constraint $T/4 < a_i < T/2$.
In the following discussions, we assume this constraint holds.
Note that \tbf{3-Partition} is strongly NP-hard: the input size is evaluated as $\sum_{i=1}^{3n} a_i$.

    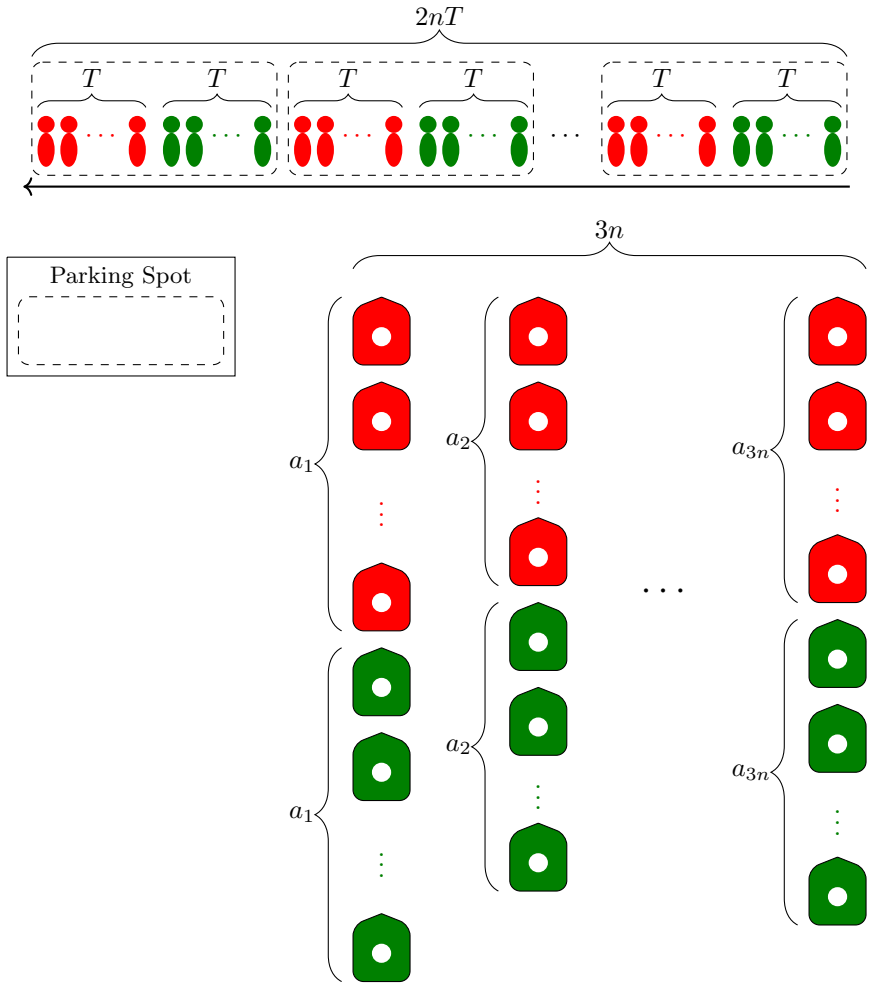
\begin{figure}[t]
        \begin{center}
            \newcommand{\human}[2]{
    \fill [#1] (#2, 0) circle (0.15);
    \fill [#1] (#2, -0.45) circle [x radius = 0.15, y radius = 0.3];
}

\newcommand{\multihumen}[3]{
    \pgfmathsetmacro{\first}{#2 + 0.4}
    \pgfmathsetmacro{\second}{#2 + 1}
    \pgfmathsetmacro{\third}{#2 + 1.6}
    \draw decorate[decoration = {brace, amplitude = 5pt}] {
        (#2 - 0.15, 0.3) -- (#2 + 1.75, 0.3)
    };
    \node (num) at (#2 + 0.8, 0.5) [above] {#3};
    \human{#1}{#2} 
    \human{#1}{\first}
    \node (dots) at (\second, -0.2) {\color{#1}$\cdots$};
    \human{#1}{\third}
}

\newcommand{\humenblock}[2]{
    \draw [dashed, rounded corners] (#1 - 2.15, 1.1) -- (#1 + 2.15, 1.1) -- (#1 + 2.15, -0.9) -- (#1 - 2.15, -0.9) -- cycle;
    \multihumen{red}{#1 - 1.9}{#2};
    \multihumen{dgreen}{#1 + 0.3}{#2};
}
\newcommand{\humenqueue}[2]{
    \draw decorate[decoration = {brace, amplitude = 10pt}] {
        (-2.15, 1.2) -- (12.15, 1.2)
    };
    \node (num) at (5, 1.6) [above] {#2};
    \humenblock{0}{#1};
    \humenblock{4.5}{#1};
    \node (dots) at (7.25, -0.2) {$\cdots$};
    \humenblock{10}{#1};
    \draw[<-,thick] (-2.3,-1.1) -- (12.2,-1.1);
}

\newcommand{\smallbus}[4]{
    \draw [fill = #1, rotate around={#4:(#2, #3)}] (#2, #3 + 0.7) [rounded corners] -- (#2 - 0.5, #3 + 0.5) [rounded corners] -- (#2 - 0.5, #3 - 0.5) [rounded corners] -- (#2 + 0.5, #3 - 0.5) [rounded corners] -- (#2 + 0.5, #3 + 0.5) -- (#2, #3 + 0.7);
    \draw [white, thick, fill = white, rotate around={#4:(#2, #3)}] (#2, #3) circle [radius = 0.15];
}

\newcommand{\busblock}[4]{
    \smallbus{#1}{#3}{#2}{0};
    \smallbus{#1}{#3}{#2 - 1.5}{0};
    \node (dots) at (#3, #2 - 1.5 - #4) {\color{#1}$\vdots$};
    \smallbus{#1}{#3}{#2 - #4 * 2 - 1.7}{0};
}
\newcommand{\busqueue}[3]{
    \draw decorate[decoration = {brace, amplitude = 10pt, mirror}] {
        (#1 -0.7, 0.7) -- (#1 -0.7, -#2 - #2 - 2.2)
    };
    \draw decorate[decoration = {brace, amplitude = 10pt, mirror}] {
        (#1 -0.7, -#2 * 2 - 2.5) -- (#1 -0.7, -#2 * 4 - 5.4)
    };
    \node (num) at (#1 -1, -#2 - 0.75) [left] {#3};
    \node (num) at (#1 -1, -#2 * 3 - 3.95) [left] {#3};
    \busblock{red}{0}{#1}{#2};
    \busblock{dgreen}{- #2 - #2 - 3.2}{#1}{#2};
}

\begin{tikzpicture}[scale=0.75]
    \humenqueue{$T$}{$2nT$}
\end{tikzpicture}
\vspace{0.5em}

\begin{tikzpicture}[baseline = -18em,scale=0.75]
    \draw (0, 0) -- (0, -2.1) -- (4, -2.1) -- (4, 0) -- (0, 0);
    \node [below] (title) at (2, 0) {\small Parking Spot};
    \foreach \x in {0.7}
        \draw [dashed, rounded corners] (0.2, -\x) -- (3.8, -\x) -- (3.8, -\x - 1.2) -- (0.2, -\x - 1.2) -- cycle;
\end{tikzpicture}  
\hspace{1em}
\begin{tikzpicture}[baseline = -15em,scale=0.75]
    \draw decorate[decoration = {brace, amplitude = 10pt}] {
        (-0.5, 1.2) -- (8.5, 1.2)
    };
    \node (num) at (4, 1.9) {$3n$};
    \busqueue{0}{1.5}{$a_1$};
    \busqueue{2.75}{1.1}{$a_2$};
    \node (dots) at (5, -4.5) {\Large$\cdots$};
    \busqueue{8}{1.25}{$a_{3n}$};
\end{tikzpicture}
            \caption{Reduction from \tbf{3-Partition} to $\mbf{BusOut}(\mcal{B}(1, 2, \{1\}))$}
            \label{fig:121_game_init}
        \end{center}
    \end{figure}
    Figure~\ref{fig:121_game_init} illustrates our reduction from \tbf{3-Partition}.
    The congestion graph $G_M$ consists of $3n$ disjoint directed paths, where the $i$-th path consists of $a_i$ red buses followed by $a_i$ green buses.
    The passenger queue $Q_M$ consists of $n$ segments each has $T$ red passengers followed by $T$ green passengers.
    The initial parking spot is empty.
    Obviously this is a polynomial-time reduction.
    
    We first show that if $M$ has a solution $\{ \{a_{p_1}, a_{q_1}, a_{r_1}\}, \{a_{p_2}, a_{q_2}, a_{r_2}\}, \dots, \{a_{p_n}, a_{q_n}, a_{r_n}\}\}$, then the configuration $(G_M,Q_M,(\varepsilon))$ is solvable.
    The following two steps remove the passengers in the $i$-th segment for $i = 1, 2, \dots, n$:
    \begin{enumerate}
        \item Dispatch every red bus from the $p_i, q_i, r_i$-th directed paths in $G_M$, which $a_{p_i} + a_{q_i} + a_{r_i} = T$ red passengers board.
        \item Dispatch every green bus from the $p_i, q_i, r_i$-th paths, which $a_{p_i} + a_{q_i} + a_{r_i} = T$ green passengers board.
    \end{enumerate}

    Next, suppose $(G_M,Q_M,(\varepsilon))$ is solvable.
    Consider the moment when the last passenger in the first segment is boarding.
    Let $i_1, i_2, \dots, i_k$ be the bus path indices from which at least one green bus has been dispatched by this moment.
    To load the $T$ green passengers in the first segment, we require 
    \[
    a_{i_1} + a_{i_2} + \dots + a_{i_k} \geq T
    \,.\]
    To dispatch a green bus in the $i$-th bus path, all the $a_i$ red buses in front must have already been dispatched, and they should have departed the station with red passengers.
    Since there are no more than $T$ red passengers, at most $T$ red buses can depart.
    Therefore,
    \[
        a_{i_1} + a_{i_2} + \dots + a_{i_k} \leq T
    \,.\]    
    Thus, $a_{i_1} + a_{i_2} + \dots + a_{i_k} = T$ and $k=3$ due to $T/4 < a_i < T/2$ for all $i$.
    This equation implies that when every passenger in the first segment has left, all and only buses in the three paths $i_1,i_2,i_3$ have gone.
    The obtained configuration is the one reduced from $M$ removing $a_{i_1},a_{i_2},a_{i_3}$.
    Hence, by repeating this argument, we obtain a solution for $M$.
%
%
%
%
\end{proof}
Theorem~\ref{theorem:s2V_comp} below generalizes Theorem~\ref{theorem:121_comp} to an arbitrary parking spot number $s$ and an arbitrary capacity set $V$.
\begin{lemma}\label{lem:duplication}
    For any $d \ge 1$ and any $(G,Q,S) \in \mcal{B}(s,c,\{1\})$, one can construct an instance $(G_d,Q_d,S_d) \in \mcal{B}(s,c,\{d\})$ in polynomial-time such that $(G,Q,S)$ is solvable if and only if so is $(G_d,Q_d,S_d)$.
\end{lemma}
\begin{proof}
    The congestion graph $G_d$ is identical to $G$ except the capacity of the buses.
    Every passenger of $Q$ is duplicated $d$ times in $Q_d$.
    The remained capacity of each element of $S$ is multiplied by $d$ in $S_d$.
    It is easy to confirm that this does not affect the (in)solvability.
\end{proof}
\begin{theorem}\label{theorem:s2V_comp}
    The problem $\mbf{BusOut}(\mcal{B}(s,c,V))$ is NP-complete for any integers $s \ge 1$, $c \ge 2$, and any nonempty capacity set $V \subseteq \mbb{N}_+$.
\end{theorem}
\begin{proof}
    We show the NP-hardness of $\mbf{BusOut}(\mcal{B}(s,2,\{1\}))$.
    This implies that $\mbf{BusOut}(\mcal{B}(s,2,\min V))$ is NP-hard by Lemma~\ref{lem:duplication}, and then so is $\mbf{BusOut}(\mcal{B}(s,c,V))$ with $c \ge 2$.
    We construct $(G_{M,s},Q_{M,s},\varepsilon^s)$ based on $(G_M,Q_M,(\varepsilon))$ in the proof of Theorem~\ref{theorem:121_comp} by duplicating the passengers and buses $s$ times as illustrated in Figure~\ref{fig:s21_game_init}.
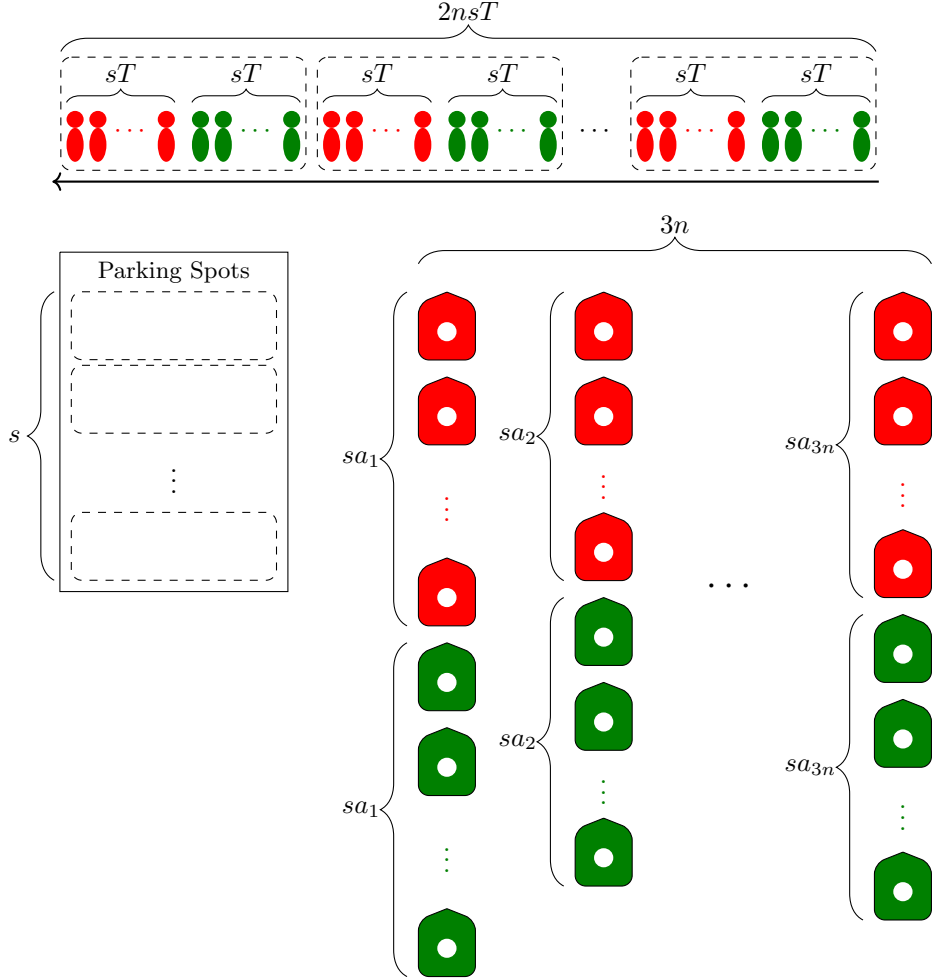
\begin{figure}[t]
    \begin{center}
        \newcommand{\human}[2]{
    \fill [#1] (#2, 0) circle (0.15);
    \fill [#1] (#2, -0.45) circle [x radius = 0.15, y radius = 0.3];
}

\newcommand{\multihumen}[3]{
    \pgfmathsetmacro{\first}{#2 + 0.4}
    \pgfmathsetmacro{\second}{#2 + 1}
    \pgfmathsetmacro{\third}{#2 + 1.6}
    \draw decorate[decoration = {brace, amplitude = 5pt}] {
        (#2 - 0.15, 0.3) -- (#2 + 1.75, 0.3)
    };
    \node (num) at (#2 + 0.8, 0.5) [above] {#3};
    \human{#1}{#2} 
    \human{#1}{\first}
    \node (dots) at (\second, -0.2) {\color{#1}$\cdots$};
    \human{#1}{\third}
}

\newcommand{\humenblock}[2]{
    \draw [dashed, rounded corners] (#1 - 2.15, 1.1) -- (#1 + 2.15, 1.1) -- (#1 + 2.15, -0.9) -- (#1 - 2.15, -0.9) -- cycle;
    \multihumen{red}{#1 - 1.9}{#2};
    \multihumen{dgreen}{#1 + 0.3}{#2};
}
\newcommand{\humenqueue}[2]{
    \draw decorate[decoration = {brace, amplitude = 10pt}] {
        (-2.15, 1.2) -- (12.15, 1.2)
    };
    \node (num) at (5, 1.6) [above] {#2};
    \humenblock{0}{#1};
    \humenblock{4.5}{#1};
    \node (dots) at (7.25, -0.2) {$\cdots$};
    \humenblock{10}{#1};
    \draw[<-,thick] (-2.3,-1.1) -- (12.2,-1.1);
}

\newcommand{\smallbus}[4]{
    \draw [fill = #1, rotate around={#4:(#2, #3)}] (#2, #3 + 0.7) [rounded corners] -- (#2 - 0.5, #3 + 0.5) [rounded corners] -- (#2 - 0.5, #3 - 0.5) [rounded corners] -- (#2 + 0.5, #3 - 0.5) [rounded corners] -- (#2 + 0.5, #3 + 0.5) -- (#2, #3 + 0.7);
    \draw [white, thick, fill = white, rotate around={#4:(#2, #3)}] (#2, #3) circle [radius = 0.15];
}

\newcommand{\busblock}[4]{
    \smallbus{#1}{#3}{#2}{0};
    \smallbus{#1}{#3}{#2 - 1.5}{0};
    \node (dots) at (#3, #2 - 1.5 - #4) {\color{#1}$\vdots$};
    \smallbus{#1}{#3}{#2 - #4 * 2 - 1.7}{0};
}
\newcommand{\busqueue}[3]{
    \draw decorate[decoration = {brace, amplitude = 10pt, mirror}] {
        (#1 -0.7, 0.7) -- (#1 -0.7, -#2 - #2 - 2.2)
    };
    \draw decorate[decoration = {brace, amplitude = 10pt, mirror}] {
        (#1 -0.7, -#2 * 2 - 2.5) -- (#1 -0.7, -#2 * 4 - 5.4)
    };
    \node (num) at (#1 -1, -#2 - 0.75) [left] {#3};
    \node (num) at (#1 -1, -#2 * 3 - 3.95) [left] {#3};
    \busblock{red}{0}{#1}{#2};
    \busblock{dgreen}{- #2 - #2 - 3.2}{#1}{#2};
}

\begin{tikzpicture}[scale=0.75]
    \humenqueue{$sT$}{$2nsT$}
\end{tikzpicture}
\vspace{0.5em}

\begin{tikzpicture}[baseline = -18em, scale=0.75]
    \draw (0, 0) -- (0, -6) -- (4, -6) -- (4, 0) -- (0, 0);
    \node [below] (title) at (2, 0) {\small Parking Spots};
    \foreach \x in {0.7, 2, 4.6}
        \draw [dashed, rounded corners] (0.2, -\x) -- (3.8, -\x) -- (3.8, -\x - 1.2) -- (0.2, -\x - 1.2) -- cycle;
    \node (dots) at (2, -3.9) {$\vdots$};
    \draw decorate[decoration = {brace, amplitude = 10pt, mirror}] {
        (-0.1, -0.7) -- (-0.1, -5.8)
    };
    \node (num) at (-0.8, -3.25) {$s$};
\end{tikzpicture}  
\hspace{1em}
\begin{tikzpicture}[baseline = -15em, scale=0.75]
    \draw decorate[decoration = {brace, amplitude = 10pt}] {
        (-0.5, 1.2) -- (8.5, 1.2)
    };
    \node (num) at (4, 1.9) {$3n$};
    \busqueue{0}{1.5}{$sa_1$};
    \busqueue{2.75}{1.1}{$sa_2$};
    \node (dots) at (5, -4.5) {\Large$\cdots$};
    \busqueue{8}{1.25}{$sa_{3n}$};
\end{tikzpicture}
        \caption{Reduction from $\mathbf{3\textrm{-}Partition}$ to $\mathbf{BusOut}(\mcal{B}(s, 2, \{1\}))$}
        \label{fig:s21_game_init}
    \end{center}
\end{figure}
    That is, $G_{M,s}$ consists of directed paths with $s a_i$ red buses followed by $s a_i$ green buses for each $i \in \{1,\dots,3n\}$, and
    $Q_{M,s}$ is $(\qr{s T}\qg{s T})^n$.
    If $M$ has a solution, the same strategy as in the proof of Theorem~\ref{theorem:121_comp} gives a solution for $(G_{M,s},Q_{M,s},\varepsilon^s)$.
    Note that just one parking spot is enough for this solution.

    Conversely, suppose $(G_{M,s},Q_{M,s}, \varepsilon^s)$ is solvable.
    We show that the extra $s-1$ spots does not make it easier to find a solution for this configuration.
    Consider the moment when the last passenger in the first segment is boarding.
    To load the $sT$ green passengers in the first segment, suppose we use some green buses from $k$ paths in $G_{M,s}$, whose indices are $i_1, i_2, \dots, i_k$.
   We must have $s(a_{i_1} + a_{i_2} + \dots + a_{i_k}) \geq sT$, i.e.,
   \[
    a_{i_1} + a_{i_2} + \dots + a_{i_k} - T  \geq 0
   \,.\]
   To dispatch a green bus in the $i$-th bus path, all the $s a_i$ red buses in front must be dispatched beforehand.
   In addition, to make a space for the green bus for the last green passenger in the segment, all \emph{but at most $s-1$} of those red buses should have left with red passengers.
   Since there are no more than $sT$ red passengers, at most $sT$ red buses can leave the station.
   Hence,
   $s(a_{i_1} + a_{i_2} + \dots + a_{i_k}) - (s-1) \leq sT$, i.e.,
   \[
    a_{i_1} + a_{i_2} + \dots + a_{i_k} - T \le 1 - \frac{1}{s} < 1\,.
   \]
   Since $a_{i_1}, a_{i_2}, \dots, a_{i_k}$ and $T$ are integers, 
   \[
    a_{i_1} + a_{i_2} + \dots + a_{i_k} - T \leq 0
   \]
   holds.
    Therefore, we have $k=3$ and $a_{i_1}+a_{i_2}+a_{i_3} = T$ as desired.
\end{proof}

We remark that the congestion graphs used in the proofs of Theorems~\ref{theorem:121_comp} and~\ref{theorem:s2V_comp} consist of disjoint directed paths.
This shows that to make $\mbf{BusOut}$ hard, no complicated graphs are required even with two colors and a singleton capacity set.

Recall that one can buy parking spots in the real smartphone game for solving instances easily.
However, it is hard to approximate the number of extra spots needed to make an instance solvable.
\begin{corollary}\label{cor:approx}
    Unless $\mrm{P}=\mrm{NP}$,
    no integer $r \ge 1$ admits a polynomial-time algorithm that computes an integer $\tilde{s}$ for a given $(G,Q)$ such that $s_0 \le \tilde{s} \le rs_0$ for the least integer $s_0$ for which $(G,Q,\varepsilon^{s_0})$ is solvable.
\end{corollary}
\begin{proof}
    Recall our reduction used in the proof of Theorem~\ref{theorem:s2V_comp}.
    Consider $s$ for $(G_{M,r},Q_{M,r},\varepsilon^{s})$ to be solvable.
    If $M$ has a solution, then $s=1$ is enough.
    Otherwise, $s=r$ is not enough.
\end{proof}
Now we turn our attention to an even restricted class of congestion graphs: namely, independent sets.
Let $\mcal{B}_\tsc{ind}(s,c,V) \subseteq \mcal{B}(s,c,V)$ be the collection of configurations where the congestion graphs have no edges.
\begin{theorem}
    Fix positive integers $s$, $c$, and $v$, and let $\mcal{A}(s,c,v) = \bigcup_{|V| \le v} \mcal{B}_\tsc{ind}(s,c,V)$.
    The problem $\mbf{BusOut}(\mcal{A}(s,c,v))$ is decidable in polynomial time.
\end{theorem}
\begin{proof}
    It is enough to show that any instance $(G,Q,S) \in \mbf{BusOut}(\mcal{A}(s,c,v))$ has at most polynomially many reachable configurations.
    Let $n$ be the number of passengers, which coincides with the sum of the capacities of the buses.

    Since $G$ and succeeding congestion graphs are independent sets, they can be identified with a multiset of pairs of a color and a capacity.
    Thus, at most $n^{cv}$ different congestion graphs (modulo isomorphism) appear in reachable configurations.

    Let $d_x$ be the maximum capacity of a bus of color $x$ in the instance.
    The state of each parking spot is a pair of a color and available seats unless the spot has no bus, so there are at most $1+\sum_{x} d_x \le 1+n$ variants.
    Thus, the total number of possible spot occupancy states in the station is bounded by $(1+n)^s$.

    For each pair of a congestion graph and a spot occupancy state, at most one passenger queue is possible to form a reachable configuration.
    Therefore, we have at most $O(n^{scv})$ reachable configurations.
\end{proof}
When $s \ge c$, every instance is solvable even if bus capacities are not finitely fixed.
\begin{theorem}\label{thm:ind_triv}
    Every configuration of $\mcal{B}_\tsc{ind}(s, c, \mbb{N}_+)$ is solvable if $c \le s$.
\end{theorem}
\begin{proof}
    Reserve a parking spot for each color.
    Then, we never reach a deadlock.
\end{proof}
Theorem~\ref{thm:ind_comp} contrasts with Theorem~\ref{thm:ind_triv}.
\begin{theorem}\label{thm:ind_comp}
    The problem $\mathbf{BusOut}(\mcal{B}_\tsc{ind}(s, c, \mbb{N}_+))$ is NP-complete if $s < c$.
\end{theorem}
\begin{proof}
    For an instance $M = \{a_1,\dots,a_{3n}\}$ of \tbf{3-Partition}, we define a configuration of $\mcal{B}_\tsc{ind}(s,s+1,\mbb{N}_+)$ as follows:
    \begin{itemize}
        \item The colors are $x_0,\dots,x_s$.
        \item The congestion graph is identified with a multiset
        \\ $\{\,(x_0,a_i)^1 \mid 1 \le i \le n \,\}\cup \{\,(x_i,2)^n \mid 1\le i\le s-1\,\} \cup \{(x_s,1)^1\}$,\\ i.e., it consists of
        \begin{itemize}
            \item one bus of color $x_0$ with capacity $a_i$ for each $i = 1, \dots, 3n$,
            \item $n$ buses of color $x_i$ with capacity $2$ for each $i = 1,\dots,s-1$,
            \item $n$ buses of color $x_s$ with capacity $1$.
        \end{itemize}
        \item The passenger queue is $Q = (x_1 x_2 \dots x_{s-1} x_0^T x_s x_1 x_2 \dots x_{s-1})^n $, where $T=\frac{1}{n}\sum_{i=1}^{3n} a_i$.
        %
        \item The initial parking spots are empty.
    \end{itemize}
    This instance is clearly constructible in polynomial time.

    Suppose that $M$ has a solution $\{\{a_{p_1}, a_{q_1}, a_{r_1}\}, \ldots, \{a_{p_n}, a_{q_n}, a_{r_n}\}\}$.
    Then, the following procedure clears each passenger segment $i = 1, \ldots, n$:
    \begin{enumerate}
        \item Dispatch one bus of each color $x_1, \ldots, x_{s-1}$ to accommodate one passenger of the respective colors.
        Those buses stay in the station occupying $s-1$ parking spots in total.
        \item Dispatch the three $x_0$-colored buses with capacities $a_{p_i}$, $a_{q_i}$, and $a_{r_i}$ to carry the $T$ passengers of color $x_0$.
        This uses and frees up the last available spot thanks to $a_{p_i} + a_{q_i} + a_{r_i} = T$.
        \item Dispatch one $x_s$-colored bus to the freed spot.
        \item All the remaining passengers in the segment board on the buses in the station.
    \end{enumerate}

    Conversely, suppose the instance has a solution.
    Then, at the beginning, to load the first $s-1$ passengers, we must dispatch buses of colors $x_1,\dots,x_{s-1}$.
    Since those buses have capacity 2, they stay and occupy $s-1$ spots, so the following $T$ passengers of color $x_0$ board using just one spot.
    Suppose we use buses of capacities $a_{i_1},\dots,a_{i_k}$ to carry them.
    Here, to carry the next passenger of color $x_s$, the $x_0$-colored bus should leave there to free up the spot for a bus of color $x_s$.
    This requires $a_{i_1} + \dots + a_{i_k}=T$, with $k=3$.
     After the $T$ passengers of color $x_0$ have left, the only way to carry the $x_s$-colored passenger is to dispatch a $x_s$-colored bus.
     Then, it immediately leaves the station.
    The following passengers of color $x_1,\dots,s_{s-1}$ automatically board on their matching buses in the station, which clears all the spots.
    The same argument applies to the following segments of the queue, and we obtain a solution for $M$.
\end{proof}

\section{Conclusion and Future Work}



In this study, we defined the decision problem $\mathbf{BusOut}$ for the mobile game \textit{Bus Out} and showed that the problem is NP-complete, even under strong restrictions, whereas slight relaxations make the problem trivial.
Namely, the problem remain NP-hard when the station has only one parking spot, buses have only two colors, a single fixed capacity, and the congestion graph consists of disjoint paths.
It is still NP-hard even when traffic congestion is absent if we have two colors and possible bus capacities are not fixed.
On the other hand, if we have no less parking spots than colors, the problem with no traffic congestion is trivial.
We also showed that it is hard to approximate the least number of spots for a congestion graph and a passenger queue to make the composed configuration solvable.

Still, there remains scope for further exploration of \tbf{BusOut}.
For example, we did not discuss the complexity of the problem when the capacity set is fixed and the number of colors is unbounded.
Furthermore, extending game into a two-player competitive version would also be an interesting direction for future research.

\bibliographystyle{unsrt}
\bibliography{ref}

\end{document}